\documentclass[journal,a4paper]{IEEEtran}

\usepackage{graphicx} 
\usepackage{pdfpages}
\usepackage{amssymb,enumitem}
\usepackage[hidelinks]{hyperref}
\hypersetup{
    colorlinks=true,
    linkcolor=blue,
    filecolor=magenta,      
    urlcolor=cyan,
    pdftitle={Overleaf Example},
    pdfpagemode=FullScreen,
}
\graphicspath{ {./images/} }
\usepackage[utf8]{inputenc} 
\usepackage[T1]{fontenc}
\usepackage{url}              
\usepackage{cite}             
\usepackage{amsthm}
\usepackage[cmex10]{amsmath}  
\usepackage{mleftright}       
\mleftright                   
\usepackage{tikz}
\usepackage{amsfonts}
\usepackage{graphicx}         
\usepackage{booktabs}         

\usepackage{amsthm}

\makeatletter
\newtheorem*{rep@theorem}{\rep@title}
\newcommand{\newreptheorem}[2]{%
\newenvironment{rep#1}[1]{%
 \def\rep@title{#2 \ref{##1}}%
 \begin{rep@theorem}}%
 {\end{rep@theorem}}}
\makeatother

\newreptheorem{theorem}{Theorem}
\newreptheorem{lemma}{Lemma}

\newtheorem{theorem}{Theorem}

\newtheorem{lemma}{Lemma}

\newtheorem{corollary}{Corollary}

\newtheorem{example}{Example}
\newtheorem{remark}{Remark}

\theoremstyle{definition}
\definecolor{britishracinggreen}{rgb}{0.0, 0.26, 0.15}
\newtheorem{definition}{Definition}

\title{Codes Correcting a Single Long Duplication Error}
\author{%
  \IEEEauthorblockN{\textbf{Daniil Goshkoder}\IEEEauthorrefmark{1}\IEEEauthorrefmark{2},
                    \textbf{Nikita Polyanskii}\IEEEauthorrefmark{3},
                    and
                    \textbf{Ilya Vorobyev}\IEEEauthorrefmark{4}}\\

\IEEEauthorblockA{\IEEEauthorrefmark{1}%
                     Skolkovo Institute of Science and Technology,
                    Moscow, Russia,
                    daniilgoshkoder@mail.ru}\\
  \IEEEauthorblockA{\IEEEauthorrefmark{2}%
                     Lomonosov Moscow State University,
                    Moscow, Russia,
                    daniilgoshkoder@mail.ru}\\ 
  \IEEEauthorblockA{\IEEEauthorrefmark{3}%
                     IOTA Foundation,
                    Berlin, Germany,
                    nikita.polyansky@gmail.com} \\                    
  \IEEEauthorblockA{\IEEEauthorrefmark{4}%
                     Technical University of Munich,
                    Munich, Germany,
                    vorobyev.i.v@yandex.ru} 
                    }

\begin{document}
\maketitle
\begin{abstract}
We consider the problem of constructing a code capable of correcting a single long tandem duplication error of variable length. As the main contribution of this paper, we present a $q$-ary efficiently encodable code of length $n+1$ and redundancy $1$ that can correct a single duplication of length at least $K=4\cdot\lceil \log_q n\rceil +1$. The complexity of encoding is $O(\frac{n^2}{\log n})$ and the complexity of decoding is $O(n)$. We also present a $q$-ary non-efficient code of length $n+1$ correcting single long duplication of length at least $K = \lceil \log_q n\rceil +\phi(n)$, where $\phi(n)\rightarrow{\infty}$ as $n\rightarrow{\infty}$. This code has redundancy less than $1$ for sufficiently large $n$. Moreover, we show that in the class of codes correcting a single long duplication with redundancy $1$, the value $K$ in our constructions  is  order-optimal.
\end{abstract}

\section{Introduction} \label{sec: introduction}

One significant challenge in recent years has been the design of high-capacity, high-density data storage that supports long-term archiving. A promising solution to this problem, linked to biotechnological advancement, is the storage of digital information in synthesized strands of deoxyribonucleic acid (DNA), e.g. see~\cite{ceze2019molecular,bornholt2016dna, church2012next}. In particular, the possibility of storing data within the cells of living organisms is being explored, which introduces a new issue: data in long-term DNA-based storage can be modified over time, as described in~\cite{mundy2004origin}. This means it is necessary to find a mechanism to protect data against possible errors, including point insertions, deletions, substitutions, and tandem duplications. The correction of substitutions, insertions, and deletions has been the subject of extensive studies for many years, e.g. see~\cite{levenshtein1965binary,brakensiek2017efficient, helberg2002multiple, sima2020optimal, macwilliams1977theory, roth2006introduction}, whereas codes correcting tandem duplications have become under investigation in recent years only.

Corrupting by a tandem duplication error means taking a word, duplicating a word's subsequence, and inserting the duplicated subsequence directly after the original subsequence. An example of tandem duplication of length 3 in a DNA sequence $ACGTCA$ is $ACGTC\overline{GTC}A$, where the overlined part is a duplication part. We note that tandem duplication errors form a special type of the so-called burst insertion errors, i.e., corrupting by a $b$-length burst insertion error means taking a word and inserting into the word at some position $b$ arbitrary symbols consequently. Clearly, a burst of insertions is a more general type of error compared to duplications and it is possible to use codes correcting bursts of insertions when dealing with duplications. However, such constructions might be inefficient in terms of redundancy as duplications are rather specific errors.

We distinguish four models of tandem duplication errors: 1) duplications of fixed length $l$, 2) short duplications of length at most $k$, 3)  long duplications of length at least $K$, and 4) duplications of arbitrary length. The focus of  this paper is on the third model. Specifically, we study codes capable of correcting a single long tandem duplication. 

\subsection{Related Work}
\subsubsection{Duplications of fixed length}
 
In~\cite{lenz2019duplication}, the equivalence of codes correcting $t$ tandem duplications of length $l$ and $t$ tandem deletions of length $l$ is proved. 
The problem of correcting a single duplication is considered in~\cite{tang2020single,tang2021error,lenz2017bounds}. In~\cite{lenz2018bounds,kovavcevic2018asymptotically}, the authors presented constructions of codes correcting $t$ tandem duplications of fixed length $l$ and proved some bounds on the redundancy of such codes. From those results, it follows that the redundancy of optimal codes correcting $t$ duplications of length $l$ is $t\log n(1+o_{l,t}(1))$.

It is worth noting the problem of correcting an arbitrary number (up to infinity) of duplications of fixed length~\cite{yehezkeally2019reconstruction, leupold2005uniformly, jain2017noise} and variable length~\cite{jain2017duplication, zeraatpisheh2019construction}. Unlike the case of a fixed number of duplications, optimal codes for this problem achieve the asymptotic rate that is strictly smaller than $1$. Note that it does not follow from the considered works on duplications of fixed length $l$ that the codes described in them can correct duplications of length greater than $l$.

\subsubsection{Duplications of short length} Codes correcting a single short duplication of length at most $k$ are discussed in~\cite{nazirkhanova2020codes}. The redundancy of the construction presented in~\cite{nazirkhanova2020codes} is $k\log n(1+o_k(1))$. 

Recall that codes correcting a single burst insertion error can be used for correcting a single tandem duplication.  This approach for constructing duplication codes is shown to be efficient for short-length errors. In~\cite{schoeny2017codes}, the authors proved the equivalence of a code correcting a burst of insertions of length not exceeding $k$ and a code correcting a burst of deletions of length not exceeding $k$. Different approaches to construct burst deletion codes are discussed in~\cite{schoeny2017codes,gabrys2017codes,lenz2020optimal, bitar2021optimal}. In~\cite{bitar2021optimal}, the authors constructed an efficiently encodable code that corrects a burst of deletions of variable length at most $k$. The redundancy of that code is $\log n(1+o(1))$ when $k=o(\log n/(\log \log n)^2)$ (for example, $k=\sqrt{\log n}$), i.e., it is asymptotically optimal by~\cite{schoeny2017codes}.
\subsection{Our contribution}
To our best knowledge, no significant results have been obtained regarding codes correcting long duplications. Our contribution is summarized below:
\begin{itemize}[wide, labelwidth=1pt, labelindent=3pt]
    \item we construct a $q$-ary code correcting a single tandem duplication of length at least $K = 4\cdot\left\lceil\log_q n\right\rceil+1$, where the length of codewords is $n+1$ and the redundancy of the code is $1$.
This code construction allows for encoding of an information word of length $n$ into the codeword of length $n+1$ with complexity $O(\frac{n^2}{\log_q n})$, and decoding of a codeword into the information word with complexity $O(n)$;
\item for an arbitrary function $\phi(n)$ such that $\phi(n)\to \infty$ as $n\to \infty$, we also show the existence of a binary (non-efficient) code of sufficiently large length $n$ and redundancy $1$ capable of correcting a duplication of length at least $K=\left\lceil\log_q n\right\rceil+\phi(n)$.
\item we use results by~\cite{lenz2018bounds} and prove that a code of length $n$ correcting a duplication of length $K=\left\lceil\log_q n\right\rceil-c$ for $c>2$ has redundancy more than $1$; thereby, our constructions can correct a single duplication of order-optimal length. 
\end{itemize}

\subsection{Outline}

The paper is organized as follows. In Section~\ref{sec: preliminaries}, we introduce important definitions and notations that are used throughout the paper. We present our code construction and explain all its components in Section~\ref{sec: code construction}. The main result is formulated in Theorem \ref{th: code construction} and its proof is divided into Lemmas \ref{lem: words without duplications}-\ref{lem: operations complexity}. In Section \ref{sec: codes with redundancy 1} we compare our result with the converse bound which follows from~\cite{lenz2018bounds}, and with another code construction that does not allow for efficient encoding. Section \ref{sec: conclusion} concludes the main part of the paper.

\section{Preliminaries} \label{sec: preliminaries}

Let $\mathbb{Z}_q$ be an alphabet consisting of $q$ letters $\{0,1,\dots,q-1\}$ and let $x=(x_1,\dots,x_n)\in\mathbb{Z}_q^n$ be a word of length $n$. 
Denote the length of a word as $|x|$. The $l$-length subword of a word $x$ that starts at position $i$ is denoted as $x_{i+{[l]}}=(x_i,x_{i+1},\ldots,x_{i+l-1})$.

\begin{definition}\label{def: tandem duplication}
A \textit{tandem duplication} of length $l$ at position $i$ with $1\le i\le n-l$ in a word $x=(uvw)$ ($|x|=n$) is defined as $\tau_{i,l}(x)=(uvvw)\in\mathbb{Z}_q^{n+l}$, where $|u|=i$, $|v|=l$ and $|w|=n-l-i$. In $\tau_{i,l}(x)$, the first subword $v$ we call \textit{the left part of the duplication} and the second subword $v$ we call \textit{the right part of the duplication}.
\end{definition}

\begin{example}\label{ex: tandem duplications}
Consider the word $x=(00123312)\in\mathbb{Z}_4^8$. Then for this word examples of duplication errors are presented: $\tau_{2,2}(x)=(0012\overline{12}3312)$ and $\tau_{4,2}(x)=(001233\overline{33}12)$, where the overlined parts are duplication parts. Resulting words $\tau_{2,2}(x)$ and $\tau_{4,2}(x)$ are elements of $\mathbb{Z}_4^{8+2}=\mathbb{Z}_4^{10}$.
\end{example}

Further in this paper, for simplicity, we call tandem duplications simply duplications. With some abuse of terminology, we also call the duplicated part $v$ in a word $(uvvw)$ a duplication.

\begin{definition}\label{def: error ball}
\textit{The error ball of a word x with radius t} is the set of all words that can be reached by at most $t$ duplications, that is, 
$$B_t(x) = \{y|\ y = \tau_{i_s,l_s}(\dots(\tau_{i_1,l_1}(x))\dots), s\le t\}\footnote{The positions $i_1,\ldots,i_s$ are chosen arbitrarily, whereas the lengths of duplications $l_1,\ldots,l_s$ depend on the model of a duplication error}.$$
In the model of long duplication errors with parameter $K$, the lengths of duplications $l_1,\ldots,l_s$ are at least $K$.
\end{definition}

\begin{definition}\label{def: duplication correcting code}
A code $C\subset\mathbb{Z}_q^n$ is called a \textit{t-duplication correcting code}, if $B_t(x)\cap B_t(y)\neq\emptyset$ implies $x=y$ for any $x,y\in C$.
\end{definition}

In this paper, we consider \textit{a code correcting a single long duplication error}, that is, $t=1$ and the length of duplication $l\ge K = 4\cdot\left\lceil\log_q n\right\rceil+1$.

\begin{definition}\label{def: non-matching subword}
    Given a word $y$, a word $x$ ($|x|\le |y|$) is called a non-matching subword for $y$ if $x$ doesn't appear in $y$ as a subword.
\end{definition}

For example, words $001$ and $111$ are non-matching subwords for word $110000$, but $000$ isn't. In further discussions, by  a non-matching subword  for a word of length $n$ we refer to the non-matching subword of length $\left\lceil\log_q n\right\rceil$.

\section{Code construction} \label{sec: code construction}

Our main result is summarized in the following Theorem \ref{th: code construction}.

\begin{theorem}\label{th: code construction}
There exists an efficiently encodable and decodable code that maps an arbitrary word of length $n$ into a word of length $n+1$ and is capable of correcting a single long duplication of length not less than $K= 4\cdot\left\lceil\log_q n\right\rceil+1$. The cardinality of this code is $q^n$ and redundancy is $1$. The encoding complexity of the proposed code is $O(\frac{n^2}{\log_q n})$ and the decoding complexity is $O(n)$.
\end{theorem}

In what follows, we construct the code 
$$C_{n+1}^f=\{y\in \mathbb{Z}_q^{n+1}|\ \exists x\in \mathbb{Z}_q^n, y=f(x)\},$$ where $f$ is a mapping
$\mathbb{Z}_q^n\rightarrow \mathbb{Z}_q^{n+1}$ from Theorem \ref{th: code construction}.

The proof of Theorem \ref{th: code construction} is divided into several lemmas. The first technical statement, Lemma \ref{lem: words without duplications}, states that in case of the absence of long duplications in codewords, we can unambiguously reconstruct a codeword after corrupting it by one long duplication error. It follows from this lemma that it is sufficient to construct a code containing words of the same length without duplications of length $l\ge K$. Then, if a duplication error occurs, the original codeword can be reconstructed by deleting one duplicated part. The idea of constructing a code correcting a single long duplication is to build an injective mapping $f:\mathbb{Z}_q^n\rightarrow\mathbb{Z}_q^{n+1}$ such that for any $x\in\mathbb{Z}_q^n$ we get a word $f(x)\in\mathbb{Z}_q^{n+1}$, which does not contain long duplications. Next, we describe a procedure that produces the desired map. This procedure is a simple repetition of two steps. We prove its correctness, explain how to decode a codeword into the information word, and determine the complexity of the described encoding and decoding operations in Lemmas \ref{lem: non-matching subword}-\ref{lem: operations complexity}.

\subsection{Idea of a code construction}

\begin{lemma}\label{lem: words without duplications}
Suppose words $x,y\in \mathbb{Z}_q^n$ do not contain duplications of length $l$. Let $x'=\tau_{i_1,l}(x)\in \mathbb{Z}_q^{n+l}, y'=\tau_{i_2,l}(y)\in \mathbb{Z}_q^{n+l}$. If $x'=y'$, then $x=y$.
\end{lemma}

\begin{proof}
Let $x'=(x_1aax_2), y'=(y_1bby_2)$ and $x'=y'$, where $|a|=|b|=l$. Consider possible cases for the mutual location of subwords $aa$ and $bb$.

1) Consider the first case when the subwords $aa$ and $bb$ do not intersect. Then, when deleting one of the subwords $b$ from $y'$, we get a word $y$ containing a subword $aa$, which contradicts the condition, so this case is impossible.

2) Consider the second case when the subwords $aa$ and $bb$ intersect, and the length of their intersection does not exceed $l$. Then, when deleting the subword $a$, which does not intersect $bb$, from the word $x'$, we get the word $x$, and it contains a subword $bb$, which contradicts the condition, so this case is impossible too.

3) Consider the final case when the subwords $aa$ and $bb$ intersect, and the length of their intersection is greater than $l$. Then we are able to write $a$ and $b$ in the form $a=a_1a_2, b=a_2a_1$. We can assume that subword $aa$ starts to the left of $bb$ (otherwise we can swap $x'$ and $y'$). Hence, both words $x'$ and $y'$ contain a subword of the form $a_1a_2a_1a_2a_1$, while the parts of the words $x',y'$ to the left and to the right of the considered subword match. We can remove the subword $a_1a_2$ (the leftmost) from the word $x'$ and get the word $x$. Similarly, we can remove the subword $a_2a_1$ (the leftmost one) from $y'$ and get the word $y$. With these deletions, subwords we are considering in both words $x$ and $y$ have the same form $a_1a_2a_1$, that is, they match. Parts to the left and to the right also match, which means $x=y$.
\end{proof}

\begin{remark}
Note that Lemma~\ref{lem: words without duplications} is true for duplications of arbitrary length, not just of length at least $K$. However, in this paper, we apply it only to long duplications. 
\end{remark}
Lemma~\ref{lem: words without duplications} is a key ingredient for constructing a code correcting a single long duplication.
\begin{corollary}\label{cor: 1}
A code $C$ consisting of words of a fixed length that do not contain duplications of length at least $K$ can correct a single duplication of length at least $K$. 
\end{corollary}
\begin{proof}
It is sufficient to check Definition \ref{def: duplication correcting code}. Let $x$, $y\in C$, $|x|=|y|=n$ and let $B_1(x)\cap B_1(y)\neq\emptyset$. So there are words $x'\in B_1(x)$ and $y'\in B_1(y)$ such that $x'=y'$. Then by Lemma~\ref{lem: words without duplications} $x=y$, that is $C$ is the code correcting a single long duplication.
\end{proof}

Note that the absence of duplications in the codeword is not a necessary condition. In the remainder of this section, we show how to encode information words of length $n$ into words of length $n+1$ without duplications of length at least $K=4 \cdot\left\lceil\log_q n\right\rceil+1$.

\subsection{Encoding into words without long duplications}

In this subsection, we construct an encoding mapping, i.e., an injective mapping $f:\mathbb{Z}_q^n\to\mathbb{Z}_q^{n+1}$ such that for any $x\in \mathbb{Z}_q^n$, $f(x)$ does not contain a duplication of length $l$, where $l$ is at least $K=4\cdot\left\lceil\log_q n\right\rceil+1$. Note that a similar approach is considered in \cite{elishco2021repeat}. The difference is that in \cite{elishco2021repeat}, the authors constructed a mapping that modifies an arbitrary word into a $k$-repeat free word, that is, without repeating subwords of length $k$. Among the differences in the algorithms, we can note that in our algorithm we use the idea of adding non-matching subwords to avoid the appearance of new duplications, which is not present in \cite{elishco2021repeat}.

Suppose that $x$ is an information $q$-ary word of length $n$. First, we append the letter $0$ to the end of $x$ and obtain the word $x_0$ of length $n+1$. Then we break our further encoding procedure into a sequence of alternating steps. On odd steps, we apply the mapping $E_1: \mathbb{Z}_q^{n+1}\rightarrow \mathbb{Z}_q^{< n+1}$, where $\mathbb{Z}_q^{< n+1}=\bigcup\limits_{i=1}^{n}\{0,1,\dots,q-1\}^i$. This mapping removes the leftmost duplication from the current word and  adds to the end of the word some data, which can be used for decoding. On even steps, we apply the mapping $E_2:\{0,1,\dots,q-1\}^{< n+1}\rightarrow \{0,1,\dots,q-1\}^{n+1}$, which allows us to keep the length of the resulted word equal to $n+1$ and avoid duplications in the right part of the word  which corresponds to the data used for decoding. Let $E: \{0,1,\dots,q-1\}^{n+1}\rightarrow \{0,1,\dots,q-1\}^{n+1}$ denote a mapping that is a composition of the mappings $E_1$ and $E_2$, that is, $E(x_0)=E_2(E_1(x_0))$. Then we show that by applying the mapping $E$ repeatedly to word $x_0$, we get a word $y$ satisfying the requirements, that is, there are no long duplications in $y$.

\subsubsection{Odd step}

Recall that in the beginning, we append the letter $0$ to the end of the word $x$ and denote it as $x_0$. Suppose that the leftmost long duplication in $x_0$ starts at position $i_1$ and this duplication has length $l_1$. Then we remove it from $x_0$ and append to the end of $x_0$ $q$-ary numbers that represent $i_1$ and $l_1$ in base two. This requires $2\cdot\left\lceil\log_q n\right\rceil$ bits ($\left\lceil\log_q n\right\rceil$ for each number, since $i_1$ and $l_1$ both do not exceed $n$). Denote the corresponding $q$-ary expressions as $i_1^{(q)}$ and $l_1^{(q)}$. We also add the letter $1$ after $l_1^{(q)}$, which indicates that this is the end of the data block about one duplication. At the same time, the length of our word has decreased by $l_1-2\cdot\left\lceil\log_q n\right\rceil-1$, where $l_1$ is the length of the leftmost duplication.
Thus, we have constructed a word $E_1(x_0)$ of length $n+1-(l_1-2\cdot\left\lceil\log_q n\right\rceil-1)$, and this length is less than $n+1$ when $l_1\ge 2\cdot \left\lceil\log_q n\right\rceil +1$. Schematically, the result of this function for a word $x_0$ is 
$$
E_1(x_0)=(D(x_0),\underbrace{i_1^{(q)},l_1^{(q)},1}_{\text{data block}}),
$$
where $D(x_0)$ is a word $x_0$ with removed duplication part.

\subsubsection{Even step} 

Our goal is to construct a mapping $E_2$ that outputs  a word of length $n+1=|x_0|$ such that no extra duplications are created. The length of the word $E_1(x_0)$ equals $n+1-(l_1-2\cdot\left\lceil\log_q n\right\rceil-1)$. Next, we modify the data block appended to the right such that the resulting word has length $n+1$. Define the number 
$$r_1=\left\lfloor \frac{l_1-2\cdot\left\lceil\log_q n\right\rceil-1}{\left\lceil\log_q n\right\rceil}\right\rfloor.$$ 
Then we get the equality $l_1=2\cdot\left\lceil\log_q n\right\rceil+r_1 \cdot\left\lceil\log_qn\right\rceil + 1 + t_1$, where $t_1$ is the remainder of dividing $l_1-1$ by $\left\lceil\log_q n\right\rceil$, $0\le t_1<\left\lceil\log_q n\right\rceil$. For the data block consisting of $i_1^{(q)}$ and $l_1^{(q)}$, we write after $i_1^{(q)}$ a non-matching subword (Definition \ref{def: non-matching subword}) of length $\left\lceil\log_q n\right\rceil$ for the part of the word standing to the left of the end of $i_1^{(q)}$ (including the last letter of $i_1^{(q)}$). 

Next, we write non-matching subwords of length $\left\lceil\log_q n\right\rceil$, but the selection of a non-matching subword is made for the part of the word that is to the left of the end of the last non-matching subword that we added. We repeat this step $r_1-2$ more times. This means that $r_1-1$ non-matching subwords are added altogether up to this moment. Then we add $t_1$ arbitrary letters from $\mathbb{Z}_q$, for example, we write $t_1$ zeros written as $0^{t_1}$. Finally, we add one non-matching subword again, corresponding to the entire part of the word standing to the left of the end of $0^{t_1}$. As a result, the length of the first data block becomes equal to $2\cdot\left\lceil\log_q n\right\rceil+r_1\cdot\left\lceil\log_q n\right\rceil+1+t_1$ which equals $l_1$. Schematically, the result of this function for a word $x_0$ is 
$$
E_2(E_1(x_0))=(D(x_0),\underbrace{i_1^{(q)}, u_1^{(q)}, l_1^{(q)},1}_{\text{data block}}),
$$
where $u_1^{(q)}$ is the part of the added data with $r_1$ non-matching subwords and part $0^{t_1}$, i.e. $u_1^{(q)}=(u_{1,1},\dots, u_{1, r_1-1}, 0^{t_1}, u_{1,r_1})$. For this step to be correct, it is necessary to explain whether there exists at least 1 non-matching subword corresponding to our conditions each time a subword is added. We prove this statement in the next subsection.

\subsubsection{Resulting mapping}  The mapping $E$ is first applied to the information word $x_0$. Then we iteratively apply this mapping to the previous output to obtain a new output until there are no duplications in the final output word. The proof that the termination happens is given in the correctness part of this section. In other words, by repeatedly applying the mapping $E$ to the word $x_0$, the process necessarily ends, and the final output word $y_s=E^s(x_0)$  for some $s$ does not contain any duplications, where $E^s(x_0)=E\circ E\circ\dots\circ E(x_0)$ is the $s^{\text{th}}$ functional power of $E$, that is, $E^s=E\circ E^{s-1}$. We define encoding mapping $f$ as $f(x)=y$ where $y$ is the resulting word $y=y_s$.

\subsection{Correctness}

Correctness of encoding follows from Lemma \ref{lem: non-matching subword} and Lemma \ref{lem: encoding terminates} below. For further discussion, we introduce the notation corresponding to the $j^{\text{th}}$ application of the mapping $E$. Let the duplication for the $j^{\text{th}}$ mapping application begin with the number $i_j$ and has length $l_j$. Similarly, we define the number 
$$r_j=\left\lfloor\frac{l_j-2\cdot\left\lceil\log_q n\right\rceil-1}{\left\lceil\log_q n\right\rceil}\right\rfloor,$$
and therefore $l_j=2\cdot\left\lceil\log_q n\right\rceil+r_j \cdot\left\lceil\log_q n\right\rceil + 1 + t_j$, where $t_j$ is the remainder of dividing $l_j-1$ by $\left\lceil\log_q n\right\rceil$, $0\le t_j<\left\lceil\log_q n\right\rceil$.

\begin{lemma}\label{lem: non-matching subword}
When applying the mapping $E_2$, there exists at least one non-matching subword satisfying our conditions each time a subword is added.
\end{lemma}

\begin{proof}
Note that the total number of $q$-ary words of length $\left\lceil\log_q n\right\rceil$ is $q^{\left\lceil\log_q n\right\rceil}\ge n$. Also, note that the number of subwords of length $\left\lceil\log_q n\right\rceil$ in an arbitrary word of length $m\ge\left\lceil\log_q n\right\rceil$ is $m-\left\lceil\log_q n\right\rceil+1$. Therefore, the number of possible non-matching subwords for a word of length $m$ is not less than $n-(m-\left\lceil\log_q n\right\rceil +1)$. It follows from this inequality that the number of non-matching subwords is smaller the longer the data word is. At the same time, with each adding operation at $E_2$, the length of the word increases. So it is enough to check the existence of a non-matching subword only at the end of the $E_2$ mapping: while searching for a non-matching subword after adding $0^{t_j}$. At this point, the length of the left part of the word is $n+1-(\left\lceil\log_q n\right\rceil+\left\lceil\log_q n\right\rceil+1)=n-2\cdot\left\lceil\log_q n\right\rceil$ (the entire length minus the length of the non-matching subword, the length of the data $l_j$ and a letter $1$), and the number of non-matching subwords is not less than $n-(n-2\cdot\left\lceil\log_q n\right\rceil-\left\lceil\log_q n\right\rceil +1)=3\cdot\left\lceil\log_q n\right\rceil-1 \ge 1$ (for $n\ge 2$). This means the existence of a non-matching subword with the required condition at any moment of $E_2$ mapping.
\end{proof}

\begin{lemma}\label{lem: encoding terminates}
For any $q$-ary word $x_0\in\mathbb{Z}_q^{n+1}$, there exists a non-negative $s\in\mathbb{Z}$ such that $y_s=E^s(x_0)$ does not contain long duplications, that is, the encoding mapping terminates. Moreover, the resulting word $y_s$ has the following form
$$y_s=(D^s(x_0), i_1^{(q)}, u_1^{(q)}, l_1^{(q)}, 1, \dots, 1, i_s^{(q)}, u_s^{(q)}, l_s^{(q)}, 1).$$
Here, $D^s(x_0)$ is the remaining part of the word $x_0$ after removing $s$ duplicated parts, and $u_j^{(q)}$ is part of the added data which contains $r_j$ non-matching subwords and a block of zeroes $0^{t_j}$.
\end{lemma}

\begin{proof}
We prove by induction that for any non-negative $j\in\mathbb{Z}$ the word $y_j$ has the form 
$y_j=(D^j(x_0), i_1^{(q)}, u_1^{(q)}, l_1^{(q)}, 1, \dots, 1, i_j^{(q)}, u_j^{(q)}, l_j^{(q)}, 1),$ where $D^j(x_0)$ is the word $x_0$ with $j$ removed parts of the duplications and $u_j^{(q)}$ is the part of the added data with $r_j$ non-matching subwords and part $0^{t_j}$, i.e. $u_j^{(q)}=(u_{j,1},\dots, u_{j, r_j-1}, 0^{t_j}, u_{j,r_j})$. We also prove that in this word the right part of any duplication cannot lie entirely in part with the added data. Then we show how Lemma \ref{lem: encoding terminates} follows from this statement.

Before proceeding to the direct proof of  Lemma \ref{lem: encoding terminates}, we note that the following inequality holds:
\begin{align*}
r_j&=\left\lfloor\frac{l_j-2\cdot\left\lceil\log_qn\right\rceil-1}{\left\lceil\log_qn\right\rceil}\right\rfloor\\
&\ge\left\lfloor\frac{4\cdot\left\lceil\log_qn\right\rceil +1-2\cdot\left\lceil\log_qn\right\rceil-1}{\left\lceil\log_qn\right\rceil}\right\rfloor\\
&=\left\lfloor\frac{2\cdot\left\lceil\log_qn\right\rceil}{\left\lceil\log_qn\right\rceil}\right\rfloor=2,
\end{align*}
that is, $r_j\ge 2$.

\textit{Base of induction}: $j=1$. The fact that the word $y_1$ has the form $(D(x_0), i_1^{(q)}, u_1^{(q)}, l_1^{(q)}, 1)$ directly follows from the description of the mapping $E$. Let us prove that the right part of any duplication cannot lie entirely in part with the added data. Note that in our word the distance between $2$ adjacent non-matching subwords (from the end of the left to the beginning of the right) does not exceed $\left\lceil\log_qn\right\rceil$, since the largest distance between them is filled by $0^{t_1}$ (here we use that $r_1\ge 2$). Also, the distance from the leftmost non-matching subword to the beginning of the data block is $\left\lceil\log_qn\right\rceil$, and the distance from the rightmost non-matching subword to the end of the data block is $\left\lceil\log_qn\right\rceil+1$. Therefore, if the right part of some duplication lies in part with the added data, then it contains at least one non-matching subword because the length of this duplication part is not less than $4\cdot\left\lceil\log_qn\right\rceil+1$. From the definition of a non-matching subword, it follows that the same subword can not lie in the left part of this duplication, which contradicts the definition of duplication, therefore, such a case is impossible.

\textit{Induction step}: assume the statement to be true for some $j-1$. Now we shall prove it for $j$. By the induction hypothesis, $y_{j-1}$ equals $(D^{j-1}(x_0), i_1^{(q)}, u_1^{(q)}, l_1^{(q)}, 1, \dots, 1, i_{j-1}^{(q)}, u_{j-1}^{(q)}, l_{j-1}^{(q)}, 1).$
Since the left part of the duplication in this word lies entirely in part $D^{j-1}(x_0)$, we can assume that when deleting one part of the duplication, we delete the left part, and then we modify subword $D^{j-1}(x_0)$ into $D^j(x_0)$, but the part with the added data does not change. After adding the new data, our word has the form $y_j=(D^j(x_0), i_1^{(q)}, u_1^{(q)}, l_1^{(q)}, 1, \dots, 1, i_j^{(q)}, u_j^{(q)}, l_j^{(q)}, 1)$, which is what we need. It remains to prove that the right part of any duplication cannot lie entirely in part with the added data. For this word, the maximum distance between two adjacent non-matching subwords is $2\cdot\left\lceil\log_qn\right\rceil+1$, and it is filled with a subword of the form $(l_j^{(q)},1,i_j^{(q)})$ (here we also use that $r_j\ge 2$). But the length of the right part of any duplication is not less than $4\cdot\left\lceil\log_qn\right\rceil+1$. Hence the right part of the duplication, which lies in the added data block, must contain at least one non-matching subword. Similarly to the induction base, this case is impossible.

For the induction step, it is also worth noting that a non-matching subword remains non-matching after deleting duplications in the main part of the word $D^j(x_0)$. This is true since when deleting one of the parts of the duplication, we can get new possible non-matching subwords from the intersection of the parts of the duplication, but at the same time, no non-matching subwords cease to be them, since new subwords of length $\left\lceil\log_qn\right\rceil$ do not appear in the word.

Now we show how Lemma \ref{lem: encoding terminates} follows from the proved statement. At any step of encoding mapping $E$, the left part of the duplication lies in the part of the word $D^j(x_0)$, which means that we can assume that we delete the left part of the duplication at the next encoding step and the length of the word $D^j(x_0)$ decreases with increasing $j$. But at the same time, the length of this word cannot be less than zero, which means that our algorithm finishes its work, that is there exists a non-negative $s\in\mathbb{Z}$ such that $y_s=E^s(x_0)$ does not contain long duplications. The form of the word $y_s$ directly follows from the proven statement with $j=s$. Lemma \ref{lem: encoding terminates} is completely proved.
\end{proof}

\subsection{Duplication correction}

Let $c$ be the word from the code $C_{n+1}^f$ from Theorem \ref{th: code construction}. Let $y=\tau_{i,l}(c)$ be a duplication error in the word $c$. We can restore the word $c$ by searching for a duplication of length $l$ in the word $y$ and removing one duplication part, which follows from Lemma \ref{lem: words without duplications}. Note that for our construction, the correction of duplication can be done by a simple operation, but in general, this operation may not be so simple.

\subsection{Codeword decoding}

In this subsection, we describe the algorithm for decoding a codeword into an information word. From Lemma \ref{lem: encoding terminates} it follows that any codeword $y$ has the form $(D^s(x_0), i_1^{(q)}, u_1^{(q)}, l_1^{(q)}, 1, \dots, 1, i_s^{(q)}, u_s^{(q)}, l_s^{(q)}, 1)=y_s$ for some word $x_0$ and non-negative integer $s$. Decoding begins by reading the word from right to left. Primarily, we read the first letter from the end, and if it is equal to $1$ (when $s\neq 0$), then there is some data about duplication in the word $D^{s-1}(x_0)$. Then we read $\left\lceil\log_qn\right\rceil$ letters, which contain information about the length of the last duplication, that is, $l_s$. At encoding, we use the equality $l_s=2\cdot\left\lceil\log_qn\right\rceil+r_s \cdot\left\lceil\log_qn\right\rceil +1+t_s$, and from here we can find the length of the part of the data block, which is between $i_s^{(q)}$ and $l_s^{(q)}$ -- the number $r_s\cdot\left\lceil\log_qn\right\rceil+t_s$, which equals $l_s-2\cdot\left\lceil\log_qn\right\rceil-1$. Let us move to the left by this number of letters and then read the next $\left\lceil\log_qn\right\rceil$ letters that give us information about $i_s$. We read exactly $l_s$ letters and received information about $l_s$ and $i_s$, that is, we can reconstruct the corresponding duplication in the word $D^{s-1}(x_0)$. We delete the data block we read from the word and restore the duplication that was deleted at the encoding mapping. At the same time, the length of the word does not change and is equal to $n+1$, since the length of the deleted part is equal to the length of the duplication in accordance with the encoding algorithm.
We repeat such an operation while the last letter of the word equals $1$. When the last letter becomes equal to $0$, we delete this letter and obtain the original information word  $x$.

\subsection{Correctness of codeword decoding}

Correctness of decoding algorithm follows from Lemma \ref{lem: encoding terminates}, because after each decoding step our word has the form $y_j=(D^j(x_0), i_1^{(q)}, u_1^{(q)}, l_1^{(q)}, 1, \dots, 1, i_j^{(q)}, u_j^{(q)}, l_j^{(q)}, 1).$ It is clear from the decoding algorithm that we simply read one block of added data and restore one duplication. At the same time, when restoring the last duplication, we consider the last data block and the next letter is $0$. This means that an information word $x$ is written to the left.

\subsection{Code cardinality and redundancy}

\begin{lemma}\label{lem: cardinality of the code}
The cardinality of a code 
$$C_{n+1}^f=\{y\in \mathbb{Z}_q^{n+1}|\ \exists x\in \mathbb{Z}_q^n, y=f(x)\},$$
where $f$ is the encoding mapping $\mathbb{Z}_q^n\rightarrow \mathbb{Z}_q^{n+1}$ from Theorem \ref{th: code construction}, equals $q^n$ and the redundancy of this code equals $1.$
\end{lemma}

\begin{proof}
$f$ is an injective mapping $\mathbb{Z}_q^n\rightarrow \mathbb{Z}_q^{n+1}$, hence the cardinality of the code is equal to the cardinality of $\mathbb{Z}_q^n$, that is, $
|C_{n+1}^f|=|\mathbb{Z}_q^n|=q^n.$ Then we are able to determine the \textit{redundancy of the code}: $\eta=(n+1)-\log_q|C_{n+1}^f|=(n+1)-\log_qq^n=n+1-n=1.$
\end{proof}

\subsection{Complexity of encoding and decoding}

Before direct proof of the lemma about the complexity of encoding and decoding, let's denote the number of $q$-ary words of length $n$ with at least one duplication of length $l\ge K$ as $CBW$ (count of bad words) and prove the following lemma.

\begin{lemma}\label{lem: CBW}
The number of $q$-ary words of length $n$ with at least one duplication of length $l\ge K$ can be estimated as
$$CBW\le n\cdot q^{n}\cdot q^{1-K}.$$
\end{lemma}
\begin{proof}
Let's consider the following chain of equalities and inequalities:
\begin{align*}
CBW&\le \sum\limits_{l=K}^{n/2} q^{n-l}\cdot (n-2l+1)\\
&\le\sum\limits_{l=K}^{n/2} \frac{q^{n}\cdot n}{q^l}\\
&=q^{n}\cdot n\cdot (q^{1-K}-q^{-\frac{n}{2}})\\
&\le n\cdot q^{n}\cdot q^{1-K}.
\end{align*}

The first inequality is the sum over all possible duplication lengths from the products of the number of corresponding words ($q$-ary words in which $l$ letters corresponding to the second part of the duplication are fixed) by the number of possible options for the start of the duplication. The second inequality is obtained by estimating the values $n-2l+1$ in the sum by their maximum value $n$, and for the final result, we use the geometric progression formula and inequality that $q^{\alpha}> 0$ for any $\alpha\in\mathbb{R}$.
\end{proof}

\begin{lemma}\label{lem: operations complexity}
The encoding complexity of the code $C_{n+1}^f$ is $O(\frac{n^2}{\log_q n})$ and the decoding complexity is $O(n)$.
\end{lemma}

\begin{proof}

\textit{Complexity of encoding}

We start with the complexity of encoding. Let us estimate the complexity of one mapping $E$, as well as the number of mapping applications. For the complexity of one mapping $E$, we identify three main operations, which contribute to the complexity we estimate:

1. Duplication search.

2. Search for non-matching subwords.

3. Deletions and insertions in the word.

\textit{1. The complexity of duplication search}

An algorithm from the paper~\cite{kolpakov1999finding} allows finding duplication with complexity $O(n)$. 

\textit{2. The complexity of searching non-matching subwords}

A. Creating a $q$-ary tree

To search for non-matching subwords, we create a $q$-ary tree of $\left\lceil\log_q n\right\rceil$ height for a word $x_0$ of length $n+1$. It is built according to the following rule: $q$-ary sequences of length $\left\lceil\log_q n\right\rceil$ are written in its leaves (starting from $(0,\dots ,0,0)$, $(0,\dots ,0,1)$ and ending with $(q-1,\dots ,q-1,q-1$). Moving to the leftmost child in the tree means that the corresponding letter in all leaves of the subtree is $0$, moving to the next child means that the corresponding letter in all leaves of the subtree is $1$, and moving to the rightmost child means that this letter is $q-1$. Note, that the number of nodes in this tree does not exceed $2\cdot n$. An example of such a tree for $q=2$ and $n=4$:

\begin{center}
  \includegraphics[width=0.55\linewidth, page=1]{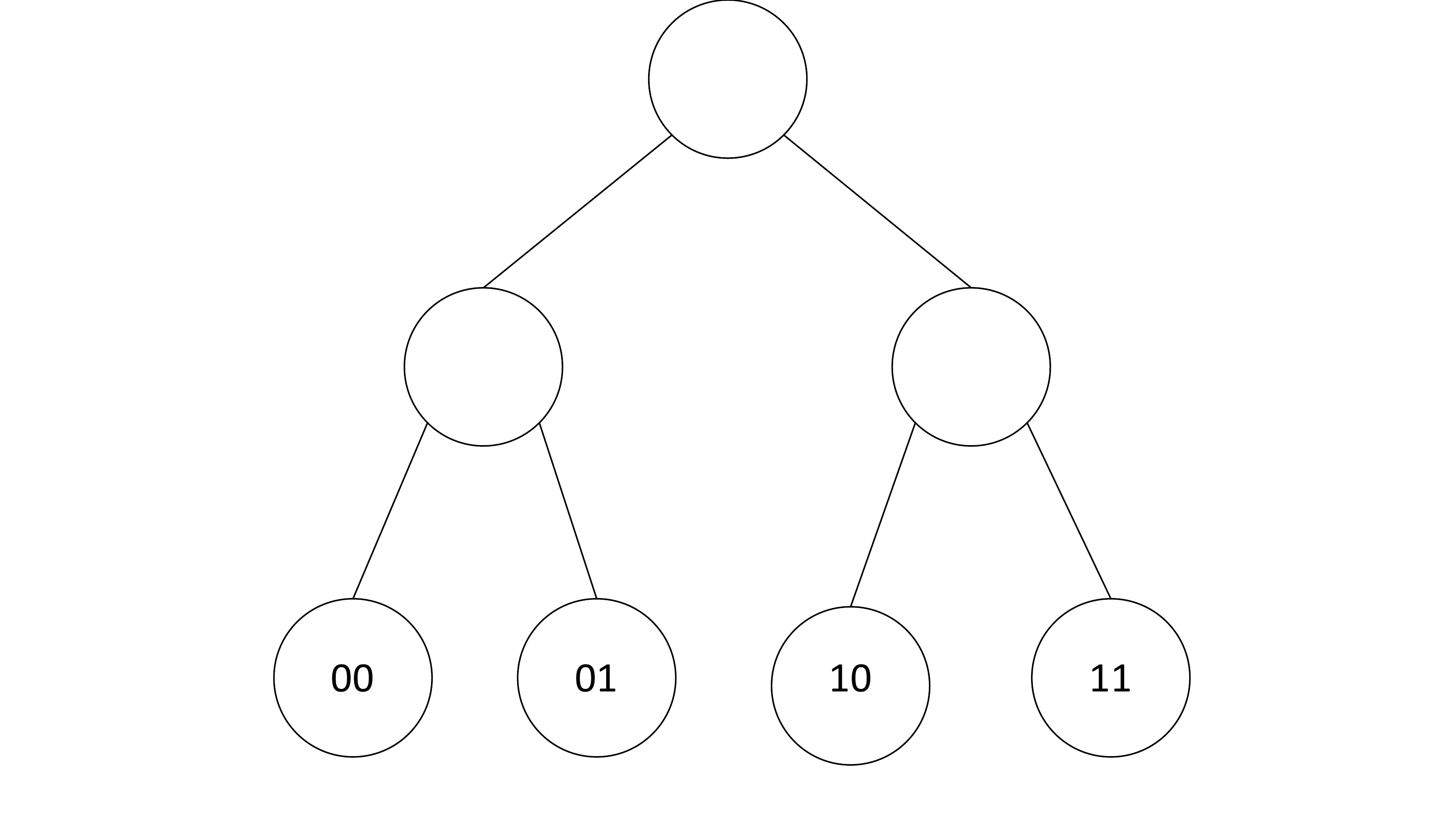}  
\end{center}

We should also add information about the word $x_0$ to our tree. To do this, add counters to the leaves equal to the number of corresponding subwords of length $\left\lceil\log_q n\right\rceil$ at word $x$. We also add counters to the internal nodes equal to the sum of the counters in its child nodes. For example, the corresponding tree for $q=2$ and the word $(0,0,1,1)$ has the following form:
\begin{center}
  \includegraphics[width=0.55\linewidth, page=2]{images/1.pdf}  
\end{center}

We create such a tree only once, before the first application of the mapping $E$. The complexity of creating such a tree is $O(n)$. Next, we will change it with transformations of our word. The operations of searching, inserting (increasing the corresponding counter), and deleting (decreasing the corresponding counter) have complexity $O(\log_q n)$, like in a $q$-ary search tree.

B. Changing the tree due to the removal of $j^{th}$ duplication and adding the data block.

This step should be done one time at each mapping. Tree changes due to the removal of the part of the duplication and the addition of the $i_j^{(q)}$ block. We remove $l_j$ letters and  decrease the counters for $l_j+\left\lceil\log_q n\right\rceil-2$ subwords. Then we add $\left\lceil\log_q n\right\rceil$ letters to the end of the word, that is, we increase the counter for $\left\lceil\log_q n\right\rceil - 1$ subwords. In total, we do $l_j +2\cdot\left\lceil\log_q n\right\rceil -3$ operations with complexity $O(\log_q n)$. At the same time $l_j\ge 4\cdot\left\lceil\log_q n\right\rceil +1$. Then the total complexity of the step is $O(l_j\cdot\log_q n)$.

C. Search for a non-matching subword in the tree

This step should be done $r_j$ times. We are looking for a leaf in the tree whose counter equals 0. We start from the root node and go to the child node whose counter is less than $\frac{q^l}{q}=q^{l-1}$, where $l=\left\lceil\log_q n\right\rceil$ is the length of the words stored in the tree (that is, the parent node of this subtree has a leaf with counter equals 0). Then we repeat the process in the same way and move to the subtree whose counter is less than $q^{l-m-1}$, where $m$ is the number of steps down the tree (being at the root of the tree $m=0$). It is necessary to prove the correctness of this algorithm.

We prove the existence of the corresponding subtree by induction by the number of steps: \textit{the base of induction} -- the total number of subwords in the tree (the sum of all the counters in the leaves) is less than $q^l$, which follows from the existence of a non-matching subword. This means that the counter of one of the child nodes is strictly less than $q^{l-1}$. \textit{Induction step} ($m+1$ step down the tree) -- we are in a subtree in which the total number of subwords (counter of this node) is less than $q^{l-m-1}$, but the counter of the node is equal to the sum of the counters in the child nodes. Then one of the counters of the child nodes is guaranteed to be less than $\frac{q^{l-m-1}}{q}=q^{l-(m+1)-1}$, which is what we need. Such a binary search has a complexity $O(\log_q n)$.

Example of a searching algorithm for $q=2$ and the word $(0,0,1,1)$:

\begin{center}
  \includegraphics[width=0.55\linewidth, page=3]{images/1.pdf}  
\end{center}

D. Changing the tree due to the addition of non-matching subwords and a block with $l_j$ and 1 to the end of the word

This step should be done $r_j + 3$ times since the addition of non-matching subwords occurs in $r_j$ quantity, as well as one changing after adding $0^{t_j}$, one changing after adding $l_j^{(q)}$ and one changing after adding $1$ to the end of the word. The complexity of insertions and deletions (increasing and decreasing counters) in the tree under consideration is $O(\log_q n)$. Each time we change the tree, we do no more than $\left\lceil\log_q n\right\rceil-1$ insertions (the lengths of $l_j^{(q)}$, $0^{t_j}$ and non-matching subwords do not exceed $\left\lceil\log_q n\right\rceil$ and we add these subwords to the end). Then the complexity of one such change does not exceed $O(\log_q n \cdot \log_q n)=O(\log_q^2 n)$.

Note that steps C and D should be done $O(r_j)$ times. Moreover, the definition of $r_j$ implies the equality $r_j=O(\frac{l_j}{\log_q n})$. The total complexity of searching non-matching subwords and tree changings for one mapping is $O(l_j\cdot\log_q n)+O(r_j\cdot\log_q n)+O(r_j\cdot\log_q^2 n)=O(l_j\cdot\log_q n)+O(\frac{l_j}{\log_q n}\cdot\log_q n)+O(\frac{l_j}{\log_q n}\cdot\log_q^2 n)=O(l_j\cdot\log_q n)$.

\textit{3. The complexity of deletions and insertions in the main word}

We store our word using AVL-tree structure (self-balancing binary search tree, see\cite{adelsonvelskii1963algorithm, foster1965information}). Each node of our tree stores three elements: the first element is the word letter, the second element is the height of its subtree and the third element is the number of nodes in the left subtree (including the parent node). The letters in the tree are stored according to the following rule. Let $N$ be an arbitrary node of our tree. Then all the nodes of its left subtree $N_{left}$ contain letters standing in the word to the left of the letter from node $N$, and all the nodes of the right subtree $N_{right}$ contain letters standing in the word to the right of the letter from node $N$. Creating such a tree has a complexity of $O(n)$ (for example, using the "middle element search" method), where $n$ is the length of the word and should be done only $1$ time. Moreover, from the tree we can return back to the "string" form of the word with complexity $O(n)$ by in-order tree traversal. Example of such a tree for $q=2$ and the word $(1,0,1,1,0,1)$:

\begin{center}
    \includegraphics[width=0.55\linewidth, page=4]{images/1.pdf} 
\end{center}

Let's describe and evaluate the complexity of the operation of searching for a letter in such a tree by the number of its position in the word. Let $i$ be the position of the letter in the word and let $j$ be the number of nodes in the left subtree of the root of the original tree (including the root of the original tree). If $i<j$, we move to the left subtree. If $i> j$, we move to the right subtree and decrease $i$ by $j$. Next, we do the same operation in the subtree to which we moved. If for some comparison it turned out that $i=j$, it means that we have found the required letter. The complexity of such an operation is determined by the height of the tree, that is, equal to $O(\log_2 n)$, where $n$ is the length of the word.

Deleting a letter from the tree is performed by searching for this letter by its number, removing it from the tree, further rebalancing the tree, if necessary, and updating the parameters stored in the nodes. The complexity of searching for a letter is $O(\log_2 n)$. The complexity of removing a node and further rebalancing in AVL-tree is $O(1)$. Finally, the complexity of updating parameters is $O(\log_2 n),$ because we should update parameters only for nodes that we passed through when searching for the removed node. Then the total complexity of deleting a letter is $O(\log_2 n)+O(1)+O(\log_2 n)=O(\log_2 n)$.

Inserting a letter can be done using an operation similar to searching for a letter in a tree. The only difference is that with $i=j$ we move to the left subtree, as with $i<j$. We continue descending the tree until we have nodes to compare the value of $i$ with. If at some point during our descent there are no nodes left, then we add a leaf with an inserted letter to the last encountered node $N_1$. The leaf is added from the side to which we had to move as a result of comparing the value of $i$ and the number of nodes in the left subtree of the node $N_1$. Then we have to rebalance a tree (with complexity $O(1)$), if necessary, and update the parameters stored in the tree nodes (with complexity $O(\log_2 n)$). The total complexity of inserting a letter equals $O(\log_2 n)$, as for deleting a letter.

Split (split an AVL-tree into two smaller AVL-trees such that all nodes in the left tree are smaller than $k$ and all values in the right tree are greater than $k$ for fixed parameter $k$) and join (join two AVL-trees such that all nodes in the left tree are smaller than those in the right tree into one balanced AVL-tree) operations are also known for AVL-trees. Their complexity is $O(\log_2 n)$.

Let us estimate the number of subword deletions and insertions from the main word at one encoding step. We do $1$ duplication deletion, $2$ insertions with duplication information ($i_j^{(q)}$ and $l_j^{(q)}$), $1$ insertion of a letter $1$ at the end of a word, $1$ insertion of subword $0^{t_j}$ and $r_j$ insertions of non-matching subwords. Thus, the number of insertions and deletions is at most $r_j+5=O(\frac{l_j}{\log_q n})$. Each inserted or deleted subword has a length not greater than $\left\lceil\log_q n\right\rceil$. Then the total complexity of insertions and deletions can be estimated as $O(\frac{l_j}{\log_q n}\cdot \log_q n\cdot{\log_2 n})=O(l_j\cdot\log_2 n)$.

\textit{4. Number of mapping applications}

In one encoding step, we reduce the length of the main word by at least $4\cdot \left\lceil\log_q n\right\rceil+1$ letters, since $K=4\cdot\left\lceil\log_q n\right\rceil+1$. Therefore, the number of encoding steps can be estimated as $O(\frac{n}{\log_q n})$.

Then we denote the overall encoding complexity $T(n)$ and estimate it, given that the number of encoding steps is $s=O(\frac{n}{\log_q n})$: 
\begin{align*}
T(n)&\le \sum\limits_{j=1}^s O(n)+\sum\limits_{j=1}^s O(l_j\cdot\log_q n)\\
&+O(n)+\sum\limits_{j=1}^sO(l_j\cdot\log_2 n)\\
&=O(n\cdot s)+\sum\limits_{j=1}^s O(l_j\cdot\log_q n)\\
&\le O\left(\frac{n^2}{\log_q n}\right)+O(n\cdot\log_q n)\\
&=O\left(\frac{n^2}{\log_q n}\right).
\end{align*}
In the presented transformation of the expression, we also take into account that the sum of the lengths of duplications does not exceed the length of the word, that is $\sum\limits_{j=1}^s l_j\le n$.
Note that the total encoding complexity is determined by $s$-fold search for duplications in a word. In all other steps, we pass information between the mapping steps, which reduces the overall complexity of these operations. Perhaps the search for duplications can also be simplified by passing information between encoding steps.

\textit{Complexity of decoding}

First, we check the last letter of the word. If it is equal to $0$, then we delete it and decoding is completed. If this letter is equal to 1, then we read $\left\lceil\log_q n\right\rceil +1$ letters from right to left, starting from the end of the word. From these letters, we find out what $l_j$ is equal to, and then we can find the number $r_j\cdot\left\lceil\log_q n\right\rceil+t_j$. Then we skip $r_j\cdot\left\lceil\log_q n\right\rceil+t_j$ letters, read the following $\left\lceil\log_q n\right\rceil$ letters, and find $i_j$. In total, we read (and delete) $l_j$ letters and we understand what duplication we need to restore since we know the beginning of the duplication $i_j$ and its length $l_j$. Then we need to restore this duplication in the word $D^j(x_0)$. 

Creating a tree from the "string" form of the word has complexity $O(n)$. Reading and deleting $l_j$ letters from the end can be implemented by splitting the tree by the $l_j$ letter from the end of the word and reading the right tree with $l_j$ nodes. The complexity of such an operation is $O(\log_2 n + l_j)$. Reading all the nodes from $i_j+1$ to $i_j+l_j$ can be implemented by splitting the tree by the node $i_j+1$, then splitting the tree by $i_j+l_j+1$ node, reading subword from the middle tree (from $i_j+1$ to $i_j+l_j$ letters of the word), and joining these three trees into one. Total complexity of this operation is $O(\log_2 n + \log_2 n + l_j + \log_2 n + \log_2 n) = O(\log_2 n + l_j)$. For inserting the reading subword we can create a described tree with this subword, split the main tree at the node with $i_j+l_j+1$ letter of the main word and join the left tree of splitting, the tree with inserting subword and the right tree of splitting. Complexity of this operation is $O(l_j + \log_2 n + \log_2 n + \log_2 n) = O(l_j + \log_2 n)$. Given that $l_j\ge K=4\cdot\left\lceil\log_q n\right\rceil +1$, the complexity of one decoding step is $O(l_j)$. Summing up the resulting expression by $j$ and adding the complexity of creating a tree, we can estimate the total decoding complexity (denote it $T_D(n)$) as:
$$T_D(n)\le O(n) + \sum\limits_{j=1}^s O(l_j)\le O(n) + O(n) = O(n).$$

\end{proof}

Lemmas \ref{lem: non-matching subword}-\ref{lem: operations complexity} imply  Theorem \ref{th: code construction}, the main result of this paper.

Note that in the average case the encoding and decoding complexity is $O(n)$. Let's denote the fraction of words of length $n$ with at least one long duplication by $p$ (the probability that an arbitrary binary word of length $n$ contains a long duplication). Then the fraction of words without duplications equals $1-p$. Also, let's denote the average complexity of encoding as $T^A(n)$ and the average complexity of decoding as $T^A_D(n)$. From Lemma \ref{lem: CBW} it follows that $p=\frac{CBW}{q^n}\le \frac{n\cdot q^n\cdot q^{1-K}}{q^n} = n^{-3}$ for $K=4\cdot\left\lceil\log_q n\right\rceil+1$. For the quantity $1-p$, we can use simple inequality $1-p\le 1$. 

If there are no long duplications in the original word after adding a letter $0$ to the end, then the encoding complexity is $O(n)$, because we add a letter $0$ to the end of the word with complexity $O(1)$ and check new word for the absence of duplications with complexity $O(n)$. If duplication is found, then the complexity of encoding is $O(\frac{n^2}{\log_q n})$. Then
\begin{align*}
T^A(n) &= O\left(n\cdot (1-p) + \frac{n^2}{\log_q n}\cdot p\right)\\
&\le O\left(n\cdot 1 + \frac{n^2}{\log_q n}\cdot \frac{1}{n^3}\right)\\
&= O(n).
\end{align*}

When decoding, we first look at the last letter of the word, so even in the best case the complexity of decoding is $O(n)$, as well as in the worst case, so the average complexity of decoding equals $O(n)$ too.

\section{Duplication codes with redundancy 1} \label{sec: codes with redundancy 1}

In this section, we compare the code $C^f_{n+1}$ with the converse bound from~\cite{lenz2018bounds} and with another (non-efficient) code construction with redundancy $1$.

\subsection{Non-efficient code construction}
Using Corollary~\ref{cor: 1}, we construct a duplication code as the set of all possible words without long duplications.
\begin{theorem}\label{th: code construction 2}
    Define the code  $C_{n+1}^0$ as
    $$
    C_{n+1}^0=\{x\in\mathbb{Z}_q^{n+1}:   x_{i+[l]}\neq x_{i+l+[l]} \ \forall i\ge 1,l\ge K \},
    $$
     i.e. the code consists of all words that do not contain duplications of length at least $K$. The code $C_{n+1}^0$  can correct a single duplication of length at least $K$. Its cardinality can be estimated as $|C_{n+1}^0|\ge q^{n+1}\cdot (1 - n\cdot q^{1-K})$.
\end{theorem}

\begin{proof}
The fact that $C_{n+1}^0$ is a code correcting a single long duplication directly follows from Lemma \ref{lem: non-matching subword} and Corollary \ref{cor: 1}.

From Lemma \ref{lem: CBW} we know that number of $q$-ary words of length $n+1$ with at least one long duplication ($CBW$) can be estimated as $CBW\le n\cdot q^{n+1} \cdot q^{1-K}$. Then the number of words without duplications
$|C_{n+1}^0|$ can be estimated as
\begin{align*}
|C_{n+1}^0|&=q^{n+1}-CBW\\
&\ge q^{n+1}-q^{n+1}\cdot n\cdot q^{1-K}\\
&=q^{n+1}\cdot (1-n\cdot q^{1-K}).
\end{align*}
Theorem \ref{th: code construction 2} is completely proved.
\end{proof}

Consider the case when $K=\left\lceil\log_q n\right\rceil+\phi(n)$, where $\phi(n)\to \infty$ as $n\to\infty$. Let $\psi(n) = \phi(n)-1$, $K=\left\lceil\log_q n\right\rceil+\psi(n)+1$. From Theorem \ref{th: code construction 2} it follows that 
$$|C_{n+1}^0|\ge q^{n+1}\cdot (1-n\cdot q^{1-K})\ge q^{n+1}\cdot (1- q^{-\psi(n)}).$$ 
Therefore, the redundancy $\eta$ can be estimated as
$$\eta\le n+1-(n+1)-\log \left(1-\frac{1}{q^{\psi(n)}}\right)=\log\frac{q^{\psi(n)}}{q^{\psi(n)}-1}.$$ 
Letting $n$ tend to infinity, we get $\lim\limits_{n\to\infty}\eta=0$. So we get that for $K=\left\lceil\log_q n\right\rceil+\phi(n)$, where $\phi(n)\to \infty$ as $n\to\infty$, there exists a code $C_{n+1}^0$ correcting a single long duplication with redundancy not more than 1. The advantage of code $C_{n+1}^f$ over code $C_{n+1}^0$ is the fact that code $C_{n+1}^f$ allows for the efficient encoding and decoding described in this paper.
\subsection{Converse bound}
Recall the result of the work~\cite{lenz2018bounds}. For optimal code $C\subset \mathbb{Z}_2^n$ correcting a single duplication of length $l$, the cardinality of the code can be estimated as $|C|\le q^{n+l+1}\cdot\frac{1}{n\cdot (q-1)}$. Then redundancy $\eta\ge n-(n+l+1 - \log_q n - \log_q (q-1)) = \log_q n + \log_q (q-1) - l - 1 \ge \log_q n-l-1$. For $l=\log_q n-c$ and $c>2$ we get that redundancy $\eta\ge c-1>1$, this means that it is impossible to construct such code with redundancy equals $1$. From this fact, it follows that we have an order-optimal length of a correctable duplication in our constructions.

\section{Conclusion} \label{sec: conclusion}

In this paper, we have discussed the problem of constructing binary codes capable of correcting a single long duplication for the $q$-ary alphabet. As the main contribution of the paper, we have presented the efficiently encodable code of length $n+1$ and redundancy $1$ correcting a duplication of length at least $4\cdot\lceil \log_q n \rceil +1$. The complexity of the encoding operation is $O(\frac{n^2}{\log_q n})$ and the complexity of the decoding is $O(n)$. In addition, we show the existence of (non-efficient) codes of length $n+1$ and redundancy $1$ correcting a duplication of length at least $\left\lceil\log_q n\right\rceil + \phi(n)$ with $\phi(n)\to\infty$ as $n\to\infty$. These results together with the known converse bound imply that the length of a correctable duplication in our constructions is order-optimal. 

Since it is already known how to deal with both a short duplication error and a long duplication error, it is natural to ask how to combine the existing approaches and build an efficient code correcting a single tandem duplication of arbitrary length. We leave this task as an open challenging problem.

\section*{Acknowledgment}
Ilya Vorobyev was supported by BMBF-NEWCOM, grant number 16KIS1005.

\newpage

\bibliographystyle{IEEEtran}
\bibliography{ref}

\end{document}